\documentclass{amsart}
\usepackage{amssymb}
\usepackage{amsthm}
\usepackage[utf8]{inputenc}
\newcommand\mut[1]{\ignorespaces}
\usepackage[usenames]{color}
\usepackage{xcolor}

\newtheorem{thm}{Theorem}
\newtheorem{lem}[thm]{Lemma}
\newtheorem{remark}[thm]{Remark}
\newtheorem{cor}[thm]{Corollary}
\newtheorem{exmp}[thm]{Example}
\newtheorem{defi}[thm]{Definition}
\newtheorem{pro}[thm]{Proposition}
\newtheorem{prob}[thm]{Open problem}

\def\gcd{\mathrm{gcd}}

\def\C{\mathbb{C}}

\def\F{\mathbb{F}}
\def\K{\mathbb{K}}

\def\x{\underline{x}}
\def\h{\underline{h}}

\renewcommand\paragraph[1]{\subsection*{#1}}

\title{On some classes  of Irreducible polynomials}

\author{Jaime Gutierrez}
\address{Departamento de Matemática Aplicada y Ciencias de la Computación\\
Universidad de Cantabria\\
Santander, Spain}
\email{jaime.gutierrez@unican.es}	
\author{Jorge Jiménez Urroz}
\address{Departamento de Matemáticas \\
Universitat Politècnica Catalunya\\
Barcelona, Spain}
\email{jorge.urroz@upc.edu}


\begin{document}


\begin{abstract}

One of the fundamental tasks of Symbolic Computation is the factorization of polynomials into irreducible factors.
\

The aim of the paper is  to produce new families of irreducible polynomials, generalizing previous results in the area. One example of our general result is that for a near-separated polynomial, i.e., polynomials  of the form $F(x,y)=f_1(x)f_2(y)-f_2(x)f_1(y)$, then $F(x,y)+r$ is always irreducible for any constant $r$ different from  zero.  

\

  We also provide the biggest known family of HIP polynomials in several variables. These are  polynomials
 $p(x_1,\ldots,x_n) \in \K[x_1,\ldots,x_n]$  over a zero characteristic  field $\K$ such that $p(h_1(x_1),\ldots,h_n(x_n))$ is irreducible over $\K$ for every $n$-tuple $h_1(x_1),\ldots,h_n(x_n)$ of non constant one variable polynomials over $\K$. The results can also be applied to fields of positive characteristic, with some modifications.
 \end{abstract}

\maketitle

\section{Introduction}

Let  $\K[x_1,\ldots,x_n]$ be the polynomial ring in $n$ variables $(x_1,\ldots,x_n)=\x$ over a  field $\K$. A polynomial $p(x_1,\ldots,x_n)=p(\x) \in \K[x_1,\ldots,x_n]=\K[\x]$ is called   a Hereditarily  Irreducible Polynomial (HIP) if  $p(h_1(x_1),\ldots,h_n(x_n))$ is irreducible in $\K[\x]$ for every $n$-tuple $h_1(x_1),\ldots,h_n(x_n)$ of non constant one variable polynomials over $\K$ , see \cite{A, RST}.
In  \cite{RST}  the authors present a class of HIP polynomials  only in two variables $x$ and $y$ over the complex number field $\C$, namely the polynomials
$p(x,y)=a(x)y+1$ such that $a(x)$ is a square free polynomial of degree at least two.
Later,   \cite{Ayad} provided an extension of this class, namely polynomials  over a zero characteristic field $\K$ of the form:
$p(x,y)=a(x)y+b(x)$  such that $a(x)$ has at least two simple roots and  $(a(x),b(x))=1$.

One of the main task of this paper is  to produce new families  of HIP.  In particular, we extend  in Section  \ref{sec:hip} the
class of the polynomials in \cite{Ayad} to  primitive polynomials  of the form $p(x,y)=a(x)c(x,y)+b(x)\in \K[x,y]$ 
such that   $c(x,y)=\sum_{i=0}^{d}c_i(x)y^i \in \K[x][y]$, $\gcd(a(x),c_d(x)b(x))=1$ and 
$a(x)$ has at least two simple roots. 

\

There are several results on the so called  difference or separated  polynomials, which are polynomials in two variables $x,y$ of the form  $P(x,y)=p(x)-p(y)$, (see the excellent  book \cite{selecta}).  The authors of  \cite{jorge} showed  that $P(x,y)+r$ is an irreducible polynomial, where $r$ is a nonzero element of a zero characteristic field $\K$. In a similar way  it is defined a near-separated polynomial as one of the form $F(x,y)=f_1(x)f_2(y)-f_2(x)f_1(y)$,  see for instance \cite{cesar,josef}.  In the last section of this paper we provide a slight generalization of the celebrated Eisenstein criteria and, with it,  we can show  that $F(x,y)+r$  is  absolutely irreducible for any  constant $r\ne 0$.   Our criterion also allows us to address other problems on irreducibility. In particular, we are able to simplify and generalize some results of \cite{mignotte}.

\

Through this paper, we denote by $\K$  an arbitrary field and by $\overline{\K}$ its algebraic closure.
  A polynomial $p(\x)\in \K[\x]$ is absolutely irreducible if it is irreducible over $\overline{\K}[\x]$.

\section{Hereditarily  Irreducible Polynomial (HIP)}
\label{sec:hip}

We start this section with a simple application of the well known Eisenstein's criterion in \cite{Ein}.

\begin{pro}\label{jorgeproposition} Let   
$a(\x)\in \K[\x]$ be with an irreducible factor of multiplicity 1. Then any  primitive polynomial $p(\x,y)=a(\x)c(\x,y)+b(\x) \in \K[\x][y] $ is absolutely irreducible, for any pair of polynomials  $c(\x,y)=\sum_{i=0}^{d}c_i(\x)y^i\in \K[\x,y]$ with  $b(\x) \in \K[\x]$ and  $\gcd(a(\x), c_d(\x)b(\x))=1$.
\end{pro}

\begin{proof} 
We take $q(\x)$ in the algebraic closure of  $\K$  to be an irreducible factor of multiplicity 1 of the polynomial $a(\x)$, which exists by hypothesis.
Consider  $p(\x,y)=\sum_{i=0}^df_i(\x)y^i$ where $f_i \in \K[\x][y]$. To prove the proposition we will apply Eisenstein criterion for $q(\x)$ in $\bar \K[\x]$. Indeed, we can apply the criterion since  $p(\x,y) $ is primitive as polynomial in $\K[\x][y]$,  $q(\x)|f_i(\x)$ for $i=1,\dots, d$, $q(\x)\nmid f_0(\x)= c_0(\x)a(\x)+b(\x)$, since $\gcd(a(\x),b(\underline x))=1$ and $q(\x)^2\nmid f_d(\x)=a(\x)c_d(\x)$ since $\gcd(a(\x), c_d(\x))=1$  and $q(\x)$ is an irreducible factor of multiplicity 1. 
\end{proof}

In order to describe  classes of HIP polynomials and after the above result it is quite natural to introduce  the following concept:

\begin{defi}
A polynomial $a(\x)\in\K[x]$ 
is called a Near  Hereditarily  Irreducible Polynomial (NHIP) if 
  $a(h_1(x_1),\ldots,h_n(x_n))$ has an irreducible factor of multiplicity $1$, for every $n$-tuple $h_1(x_1),\ldots,h_n(x_n)$ of non constant one variable polynomials.
\end{defi}

We need the following result that appears in \cite{ZurGu}:
\begin{lem}\label{lem:zurgu} Let   
$a(\x), b(\x)$ be polynomials in $\K[\x]$  with $\gcd(a(\x),b(\x))=1$. Then for every $n$-tuple $h_1(x_1), \ldots, h_n(x_n)$ of non constant one variable polynomials over $\K$ we have
$\gcd(a(h_1(x_1),\ldots,h_n(x_n) ,b(h_1(x_1),\ldots,h_n(x_n))=1$. 
\end{lem}

And consequence we have the following:

\begin{lem} \label{primitivo} Let $a(\x,y) \in \K[\x][y]$ be a primitive polynomial, then for every $n+1$-tuple $h_1(x_1), \ldots, h_n(x_n), g(y)$ of non constant one variable polynomials over $\K$ 
$p(h_1,\dots, h_n,g(y)) = p(\h,g(y))$ is also primitive in $\K[\x][y]$.
\end{lem}

\begin{proof} 

Let $p(\x,y)=\sum_{i=0}^n p_i(\x)y^i$ with $\gcd(p_0(\x),\dots,p_n(\x))=1$. Then, by previous Lemma \ref{lem:zurgu} $\gcd(p_0(\h),\dots,p_n(\h))=1$, and hence $p(\h,y)=\sum_{i=0}^n p_i(\h)y^i$ is primitive. Now suppose that $p(\h,g(y))$ is not primitive in $\K[\x][y]$. Then there exist a non constant  polynomial $q(\x) \in \K[\x]$ such that $p(\h,g(y))=q(\x)R(\x,y)$, for certain polynomial $R(\x,y) \in \K[\x,y]$. But then, taking $n+1$ different values of $g(y_l)=a_l \in \bar K $, we get $p(\h,a_l)=\sum_{i=0}^n p_i(\h)a_l^i=q(\x)R(\x,y_l)$ for $l=0,\dots, n$ and multiplying by the inverse of the Vandermonde matrix $A=(a_l^k)$ we get

$$p_i(\h)=q(\x)A^{-1}R(\x,y_l)$$

hence $q(\x)$ divides $p_i(\h)$,  for $i=0,\ldots,n$ which is a contradiction.
\end{proof}

Now, we are in the conditions to show this elementary result.

\begin{cor} \label{jorgecoro}
Let   
$a(\x)\in \K[\x]$ be a NHIP polynomial. Then any primitive  polynomial  $p(\x,y)=a(\x)c(\x,y)+b(\x) \in \K[\x][y]$ is  HIP  for any pair $c(\x,y)=\sum_{i=0}^{d}c_i(\x)y^i\in \K[\x,y]$  with $b(\x) \in \K[\x]$  and $\gcd(a(\x),c_d(\x)b(\x))=1$.
\end{cor}
\begin{proof}
The proof is an immediate consequence of  Proposition \ref{jorgeproposition} and the above result.  Observe that primitiveness is maintained thanks the above Lemma \ref{primitivo}. 

\end{proof}
Obviously, any HIP polynomial is a NHIP one.  The following interesting result 
 will help us to provide NHIP univariate polynomials and HIP polynomials.

\begin{lem}\label{jorgelema}
Let $f(x) \in \K[x]$  be  a  polynomial  with formal derivative $f'(x)\not=0$, and $a, b\in \K$ with $a\not=b$. Then the polynomial $(f(x)-a)(f(x)-b)$ has at least two simple roots in
 $\overline{\K}$.
\end{lem}
\begin{proof} 
We write $q(x)=(f(x)-a)(f(x)-b)=\displaystyle{\prod_{i=1}^{d}}(x-a_i)^{e_i}$ where $a_i \in   \overline{\K}$, for $i=1,\ldots,d$. And derivating,  we have: $$q'(x)=f'(x)(2f(x)-a-b)=h(x)g(x),$$
  where:
$$h(x)=\prod_{i=1}^{d}(x-a_i)^{e_i-1}, \quad g(x)=\sum_{i=1}^{d}e_i\prod_{i\not=j}(x-a_j)$$
We state that  $h(x)$ divides $f'(x)$.   In order to prove this, we distinguish two  cases, depending on the  characteristic of the field $\K$. 

Now,  if   the characteristic of $\K$ is $2$, then $0\not=(a+b)f'(x)=g(x)h(x)$. Otherwise, $\gcd(h(x),2f(x)-a-b)=1$ and we get $h(x)$ divides $f'(x)$.
 Indeed, note that the roots of $h(x)$ are $a_i$, the roots of $q(x)$, and since  $f(a_i)=a$ or
$f(a_i)=b$. We can not have $2f(a_i)-a-b=0$ since $a\ne b$ by hypothesis.  Then,
$$\frac{\deg q(x)}{2}-1\geq \deg f'(x)\geq \deg h(x)=\deg q(x)-d$$
It implies that $\deg q(x) \leq 2d-2$ and finishes the proof.

\end{proof}
{\bf Remark.} The previous result is an improvement of Lemma 3 of \cite{Ayad}.

\

\noindent The following example illustrates that  the hypothesis $f'(x)\not=0$   can not be omitted.
\begin{exmp} Let $f(x)=x^p \in \F_p[x]$, then
$(f(x)-1)(f(x)+1)=(x^2-1)^p$
\end{exmp}

\

An immediate consequence of Lemma \ref{jorgelema} is:

\begin{cor} If $h(x) \in \K[x]$ has $n>1$ simple roots and $f(x)\in \K[x]$  is so that $f'(x)\ne 0$, then $h(f(x))$ has at least $n$ simple roots in
$\overline{\K}$.
\end{cor}
\begin{proof} Suppose $h$ has $a_i,i=1,\dots, n$ as simple roots and $h(x)=p(x)\prod_{i=1}^n(x-a_i)$ for  $p(x)\in \bar\K[x]$. The result is trivial  if each of the factors 
$f(x)-a_i$ has at least one simple root. Otherwise, suppose $f(x)-a_1$ has no simple roots. Then, applying Lemma \ref{jorgelema} to $(f(x)-a_1)(f(x)-a_i)$ for $i=2,\dots,n$, we conclude that 
$f(x)-a_i$ must have two
simple roots and these simple roots are necessarily all distinct, so in total $h(f(x))$ would have 
$2n-2\geq n$
simple roots.
 \end{proof}

\begin{remark}  Note that in more than one variable, a polynomial can have factors of multiplicity one, 
while the composition could have all the factors with multiplicity bigger than
one, as one can see for example taking the polynomial $p(x_1,\ldots,x_n)=x_1\cdots x_n$ and considering the n-tuple  $\bar h(\x)=(x_1^ 2 ,\ldots, x_n^ 2)$ since then $p(\bar h(\x))=x_1^ 2\ldots x_n^ 2.$
\end{remark}
\

We now show  the main result of this section as consequence of Corollary \ref{jorgecoro} and Lemma \ref{jorgelema}.

\begin{thm}\label{main}  Let $\K$ be a field of characteristic zero and $p(x,y)$ be a primitive bivariate polynomial $\K[x][y]$ given by  $p(x,y) =a(x)c(x,y)+b(x)$ with  $c(x,y)= \sum_{i=0}^{d}c_i(x)y^i \in \K[x][y]$  
such that $\gcd(a(x),c_d(x)b(x))=1$, and
$a(x)$ has at least two simple roots. Then $p(x,y)$ is a HIP polynomial.
\end{thm}

\begin{proof}  First note that since we are in a field of characteristic zero, any non constant polynomial $f(x)$ has non zero derivative. Hence,  if $a(x)=(x-a)(x-b)m(x)$ where  $a,b$ are  two different simple roots, then $a(f(x))=(f(x)-a)(f(x)-b)m(f(x))$ will also have at least two simple roots by  Lemma \ref{jorgelema}. Moreover the polynomial $p(f(x),g(y)) \in\K[x][y]$ is primitive by  Lemma \ref{primitivo}. On the other hand, 
$$p(f(x),g(y)) = A(x)C(x,y)+B(x)$$ for certain polynomials  $C(x,y) = \sum_{i=0}^{D}C_i(x)y^i \in \K[x][y]$, 
$A(x)=a(f(x))$, $B(x)=b(f(x))$ and $C_D(x)= \alpha c_d(f(x))$  with $0 \not = \alpha \in \K$. Since $\gcd(a(x), c_d(x)b(x))=1$ then  by Lemma \ref{lem:zurgu} we have $\gcd(A(x), C_D(x)B(x)) =1$. 
So to get the result, we just need to apply Proposition \ref{jorgeproposition} to the polynomial $p(f(x),g(y))$. 

\end{proof}

{\bf Remark.} The result can not be generalized to fields of positive characteristic in full generality. For example $(x^2-1)y-1$ is not HIP in $\F_p$, since $(x^{2p}-1)y^p-1=((x^2-1)y)^p-1$ is reducible in $\F_p$.

\

As consequence of Theorem \ref{main} and  Corollary \ref{jorgecoro} we are able to construct HIP polynomials  in any arbitrary number
of variables. For example, by Theorem \ref{main} we take a HIP  bivariate polynomial $a(x_1,x_2) \in \K[x_1,x_2]$ then it is NHIP polynomial,  now choosing polynomials $c(x_1,x_2,y), b(x_1,x_2)$ as in Corollary \ref{jorgecoro} we get that  $a(x_1,x_2)c(x_1,x_2,y)+b(x_1,x_2) \in \K[x_1,x_2,y]$  is HIP.

\section{Eisenstein criterion and  some classes of irreducible polynomials in two variables}
In this section we deal  with the irreducibility of certain classes of polynomials in two variables.

\begin{pro}\label{teomain}  Let $\phi(x,y)=(u(x,y),v(x,y))$ be an automorphism of $\K[x,y]$ and
$F(x,y)=uQ(u,v)+r(u)$ be a primitive polynomial such that $r(0)\ne 0$  and $Q(u,v)=\sum_{j=0}^dp_j(u)v^i$  verifies  $u\nmid p_d(u)$. Then $F(x,y)$ is absolutely irreducible.
\end{pro}
\begin{proof} 
Just apply Eisenstein criterion with the prime $u\in \overline \K[u]$ to obtain the irreducibility of $F(x,y)=\tilde F(u,v)$ in the ring $\overline \K[u][v]=\overline \K[u,v]=\overline \K[x,y]$. Note that $\tilde F(u,v)\in \K[u][v]$ is primitive, and $r(0)\ne 0$ implies $u\nmid r(u)$ by Lemma \ref{lem:zurgu}.

\end{proof} 

\begin{remark} The previous result somehow gives us a generalization of Eisenstein criterion since it allows us to consider primes in the full ring $\K[x,y]$ in which we want to know irreducibility of certain polynomial. 
Somehow is like proving that the polynomial $a+bx$ is irreducible over $\K[x]$ by noticing that  $x$, prime in $\K[x]$, divides all the terms except the first, and $x^2$ does not divide the leading term.
\end{remark}
 We include now other less obvious  applications of the result. The first is a generalization of Proposition 1.5 in \cite{jorge}.

\begin{thm}  Let  $f_1(t), f_2(t) \in \K[t] $ two distinct polynomials of respective degree $d_1>0$ and $d_2 > 0$,  $0\ne r \in \K $ and $F(x,y)=f_1(x)f_2(y)-f_2(x)f_1(y)+r$. We have, 
\begin{itemize}
\item If char $\K =0 $ then   $F(x,y)$ is absolutely irreducible.
\item If char $\K=p$ and $p\nmid d_1-d_2$, then  $F(x,y)$ is absolutely irreducible.
\end{itemize} 
\end{thm}
\begin{proof} 

Observe that we can suppose $d_1\ne d_2$ since otherwise, if $d_1=d_2$, then 
$f_1(t)=af_2(t)+h(t)$ with deg$(h(t))<$deg$(f_2(t))$ and $a\in \K$, and 
$$
F(x,y)=(af_2(x)+h(x))f_2(y)-(af_2(y)+h(y))f_2(x)+r=h(x)f_2(y)-h(y)f_2(x)+r.
$$
 We write $f_l(t) =\sum_{i=0}^{d_l}a_{j,l}t^j \, (j =1, 2)$, and taking  $u=x-y$, $v=y$. Then, by Taylor expansion,  (by abuse of notation we will use $\frac{f_l^{(j)}(y)}{j!}$
for the $j$-th hasse derivative in positive characteristic), we have

$$
f_l(u+v)=\sum_{j=0}^{d_l}a_{j,l}(u+v)^j=\sum_{j=0}^{d_l}a_{j,l}\sum_{i=0}^j\binom{j}{i}u^iv^{j-i}=\sum_{i=0}^{d_l}\left(\sum_{j=i}^{d_l}a_{j,l}\binom{j}{i}v^{j-i}\right)u^i
$$
and so $F(x,y)$ can be written as a polynomial in the variables $u,v$ as 
$$
F(x,y)=\sum_{i=0}^{d_1}\left(\sum_{j=i}^{d_1}a_{j,1}\binom{j}{i}v^{j-i}\right)u^if_2(v)-\sum_{i=0}^{d_2}\left(\sum_{j=i}^{d_2}a_{j,2}\binom{j}{i}v^{j-i}\right)u^if_1(v)+r.
$$
Now observe that deg${_v}\left(\sum_{j=i}^{d_l}a_{j,l}\binom{j}{i}v^{j-i}\right)u^if_m(v)\le d_1+d_2-i$, and so we can write 
$$
F(x,y)=u(f_1'(v)f_2(v)-f_1(v)f_2'(v))+uH(v,u)+r
$$
with deg${_v}H(u,v)<d_1+d_2-1$, and 
$$
f_1'(v)f_2(v)-f_1(v)f_2'(v)=(d_1-d_2) c y^{d_1+d_2-1}+R(y)
$$ 
 for some $c\ne 0$ and  deg$(R(y))<d_1+d_2-1$. The result follows from  Proposition \ref{teomain} by choosing the automorphism $\phi=(u,v)$.

\end{proof}

\

Another  trivial application of Proposition \ref{teomain} gives the following result.   We will use the following notation: we let $H(x,y)=\sum_{i=0}^da_i(x)y^i$ be a polynomial in $\K[x,y]$ with  coefficients 
$a_i(x)=\sum_{j=0}^{d_i}a_{ij}x^j$ and consider $f(x)=\sum_{j=0}^{d_f}f_jx^j$. We will let $M=\{0\le i\le d\,:\,d_i+id_f=m\}$, where $m=\max \{d_j+jd_f\}$.

\begin{cor} \label{corauto}   Let $H(x,y)$  and $f(x)$ as above and consider   $F(x,y)=(y-f(x))H(x,y)+r$ with $0\ne r \in \K$. Then if $\sum_{i\in M} a_{id_i}(f_{d_f})^i\ne 0$, then $F(x,y)$ is absolutely irreducible.
\end{cor}

\begin{proof}
We observe that  $H(x,f(x)+u)= \sum_{j=0}^d\frac{\partial^jH(x,f(x))}{\partial y^j}\frac{u^j}{j!}$ and
 $$
\deg{\frac{\partial^jH(x,f(x))}{\partial y^j}}\le  \max_{i\le j}\{d_i+(i-j)d_f\}.
 $$
Hence, 
$
\deg_x{H(x,f(x)+u)}\le \deg{H(x,f(x))}=\max_{i\ge j}\{d_i+id_f\}
$
and  we have $H(x,f(x)+h)=\left(\sum_{i\in M} a_{id_i}(f_{d_f})^i\right)x^m+R(x)$ with $\deg{R(x)}<m.$ The result follows by applying Proposition \ref{teomain}  with $u = y-f(x)$ and $v = x$. 

\end{proof}

 It would be interesting, from the computational point  of view, to recognize whether a given polynomial $F(x,y)$ is of the form considered in Corollary \ref{corauto}.  It comes the following question:

\begin{prob}  Given $F(x,y)\in \K[x,y]$, determine wether it exists a polynomial $f(x)\in \K[x]$ and a constant $r\in \K$ such that  $F(x,y)=(y-f(x))H(x,y)+r$ for some polynomial $H(x,y)$.
\end{prob}

Observe that this is equivalent to find $f(x)\in \K[x]$ such that $F(x,f(x))$ is constant.
 
 \
 
 We intend to address this  in a future work, but here we would like to include a simple example. Suppose 
$F(x,y)=\sum_{i=0}^db_i(x)y^i$ for some polynomials $b_i(x)\in K[x]$. Note that, in case that there are $f(x)$ and $r$ in the above conditions, then by substituting $f(x)=\sum_{i=0}^df_ix^i$ into $F(x,y)$ we get
$$
r=F(x,f(x))=\sum_{i=0}^db_i(x)f(x)^i=b_0(x)+f(x)G(x),
$$
for some $G(x)\in \K[x]$. In particular, $f(x)$ divides $b_0(x)-r$, and $\deg{f(x)}\le \deg{b_0(x)}$. 
Consider $F(x,y)=-x^5-x^3y^2+x^2y+y^3+1$.  First we write it as a polynomial in $\K[x][y]$ as $F(x,y)=y^3-x^3y^2+x^2y-x^5+1$. So $\deg{f(x)}\le 5$. Now, substituting $f(x)=\sum_{i=0}^df_ix^i$ into $F(x,y)$ we see, using SAGE,  that the only possible solutions are $f(x)=x^3$ or $f(x)=\pm ix$, and $r=1$, and we get the expression $F(x,y)=(y-x^3)(y^2+x^2)+1$. 

\

It is interesting to observe that there  could be more that one expression for $F(x,y)=(y-f(x))H(x,y)+r$ with the same constant $r$ as it is the case for example for the polynomial  $F(x,y)=\prod_{i=1}^n(y-f_i(x))H(x,y)+r$. However, if $(y-f_1(x))H_1(x,y)+r_1=(y-f_2(x))H_2(x,y)+r_2$ with $r_1\ne r_2$, then $f_1(x)=f_2(x)+r$ for some constant $r$. Indeed, substituting $y$ by $f_1(x)$ inrthe previous identity  gives the desired result. 

\

The polynomials in Corollary \ref{corauto} have already been considered in the literature.
In \cite{mignotte} the authors consider the polynomial $F(x,y)=h(x)\prod_{i=1}^n(y-f_i(x))+g(x)$ and they are interested on the irreducibility over the rational function field  $\K(x)$ with certain conditions on the degree of $f_i(x)$. We can drop all those conditions in two different ways:

\begin{cor} \label{cormignote} Let $F(x,y)$  be the polynomial defined as above, with distinct $n$ polynomials $f_i(x)$ and such that  with $\sum_{i=1}^n$deg($f_i(x))>0$.

\begin{itemize}

\item[a)] If $h(x)$ has a simple root $\alpha\in\overline \K$ and $\gcd (h(x),g(x))=1$ then $F(x,y)$ is absolutely irreducible.

\item[b)] If $0\not= g(x)=r\in \K$ and $f_n(x)-f_i(x))\notin\K$, for  $i=1,\dots,n-1$,  then $F(x,y)$ is absolutely irreducible.
\end{itemize}
\end{cor}

\begin{proof}   a) is an immediate consequence of Proposition \ref{jorgeproposition}.

To prove b) we will apply Proposition \ref{teomain}  with $u=y-f_n(x),v=x$.  Indeed,

$$
F(x,y)=uh(x)\prod_{i=1}^{n-1}(u-g_i(x))+r,
$$
 where $g_i(x)=f_i(x)-f_n(x)$. Then
 $$
 \prod_{i=1}^{n-1}(u-g_i(x))=\sum_{j=0}^{n-1}p_j(x)u^j,
 $$
 where 
 $$
 p_j(x)=(-1)^{n-1-j}\sum_{1\le i_1<\dots<i_{n-1-j}\le n-1}\prod_{l=1}^{n-1-j}g_{i_l},
 $$

 in particular deg$(p_j(x))<$deg$(p_0(x))$ since 
 $$
 \left(\prod_{l=1}^{n-1-j}g_{i_l}\right)\left|\left(\prod_{l=1}^{n-1}g_{l}\right)\right.
 $$ for any set of sudindexes, and  $f_n(x)-f_i(x)\notin\K$ (for  $i=1,\dots, n-1$), 
 and hence $F(x,y)$ is primitive in $\K[u][x]$  and $u\nmid p_d(u)$.
  \end{proof}

{\bf Remark.} Observe that b) is also considered in \cite{mignotte}.

\end{document}